\documentclass[a4paper,fleqn]{cas-sc}

\usepackage[authoryear,round]{natbib}
\usepackage[english]{babel}
\usepackage{amsthm}
\usepackage{mathrsfs}
\usepackage{bm}
\usepackage{mathtools}
\usepackage{cases}
\usepackage{float}
\usepackage{algorithm}
\usepackage{algpseudocode}
\usepackage{enumerate}
\usepackage{enumitem}
\usepackage{multirow}
\usepackage{longtable}
\usepackage{threeparttable}
\usepackage{lscape}
\usepackage{rotating}
\usepackage{framed}

\hypersetup{colorlinks=true,citecolor=blue,linkcolor=blue,urlcolor=blue}

\newtheorem{theorem}{Theorem}[section]

\newtheorem{remark}{Remark}

\newcommand{\bX}{\mathbf{X}}
\newcommand{\bmu}{\boldsymbol{\mu}}
\newcommand{\bpsi}{\boldsymbol{\psi}}

\newcommand{\inner}[2]{\langle\langle #1, #2 \rangle\rangle_w}
\newcommand{\norm}[1]{|\!|\!| #1 |\!|\!|}

\begin{document}
\let\WriteBookmarks\relax
\def\floatpagepagefraction{1}
\def\textpagefraction{.001}

\shorttitle{Projection-Based Outlier Detection in Interval--Valued Functional Data}
\shortauthors{H. Xu et~al.}

\title [mode = title]{Projection-Based Outlier Detection in Interval-Valued Functional Data}

\author[1]{Hao Xu}
\ead{haoxu124@buaa.edu.cn}
\author[3, 4]{Wan Tian}
\ead{wantian61@foxmail.com}
\author[1,2]{Zhongfeng Qin}
\cormark[1]
\ead{qin@buaa.edu.cn}

\affiliation[1]{organization={School of Economics and Management, Beihang University},
                city={Beijing},
                postcode={100191},
                country={China}}
\affiliation[2]{organization={Key Laboratory of Complex System Analysis, Management and Decision (Beihang University)},
                city={Beijing},
                postcode={100191},
                country={China}}
\affiliation[3]{Advanced Institute of Information Technology, Peking University}    

\affiliation[4]{organization={Wangxuan Institute of Computer Technology, Peking University},
                city={Beijing},
                postcode={100871},
                country={China}}  

\cortext[cor1]{Corresponding author}

\begin{abstract}
Outlier detection is a fundamental task for ensuring reliable statistical
modeling and inference. Interval-valued functional data (IVFD), in which each
observation is represented by an interval-valued curve that preserves the
variability and uncertainty within the observation, have attracted increasing
attention in statistics and related applications. Developing effective outlier detection procedures for IVFD is therefore an important methodological problem. To address this issue, we develop a robust projection-based
outlier detection framework. We first represent each interval-valued
functional observation through its center and log-radius functions and apply
interval-valued functional principal component analysis (IFPCA) to obtain a
joint low-dimensional representation. We then introduce the
interval-valued least trimmed functional scores (ILTFS) method, which
identifies a robust reference subset by minimizing a trimmed aggregate of
standardized IFPCA score distances. Finally, we proposed the ILTFS-FDR outlier detection procedure by converting the resulting projection distances into empirical $p$-values and  adjusting using the Benjamini--Hochberg
procedure at a prespecified target false discovery rate level. Theoretically,
we derive the finite-sample  breakdown point of the ILTFS mean
estimator and establish the descent property of the concentration-step
algorithm. Simulation studies and an empirical application to high-frequency
ETF data demonstrate the effectiveness and robustness of the proposed ILTFS-FDR procedure in detecting abnormal interval-valued
functional observations.
\end{abstract}


\maketitle
\section{Introduction}

Over the past decade, the modeling and analysis of interval-valued data have attracted growing attention in statistics, econometrics, and machine learning \citep{han2012autoregressive,gonzalez2013constrained,sun2022model,tian2025block}.
Interval-valued data consist of observations represented by lower and upper bounds, allowing each observation to capture both its overall level and the variation within the corresponding interval. In practice, such data can be generated in several ways. First, some variables are naturally recorded in terms of their lower and upper values over a given period, as illustrated by the minimum and maximum temperatures \citep{sun2022model}.
Second, measurement error, limited observational precision, or inherent uncertainty may prevent the true value from being observed exactly and instead confine it to an interval \citep{sun2014random}. Third, individual-level point-valued observations may be aggregated across time, space, or groups to reduce storage and computational costs \citep{billard2003statistics}. Finally, sensitive attributes, such as income or age, may be recorded as intervals to protect individual privacy \citep{kang2024hausdorff}.
Compared with classical point-valued data, interval-valued observations offer two main advantages. First, they reduce the information loss caused by representing a complex observation with a single value, thereby providing a richer basis for efficient estimation and statistical inference \citep{gonzalez2013constrained,sun2022model}. Second, interval representations can accommodate part of the variation induced by measurement error, thereby reducing the influence of isolated deviant measurements on model estimates and facilitating more robust inference. When interval-valued observations are densely recorded over a continuous domain, each observation forms a complete interval-valued curve, giving rise to interval-valued functional data (IVFD). Accordingly, the existing literature can be broadly divided into two strands: the statistical analysis of general interval-valued data and the modeling and analysis of IVFD.

First, research on interval-valued data already covers a wide range of topics.
In regression analysis, early foundational models have been extended to constrained, nonlinear, linear, and robust formulations \citep{billard2000regression,limaneto2017nonlinear,sun2016linear,zhao2023robust}.
For temporally ordered interval observations, forecasting and dynamic modeling of interval-valued time series have been developed, including exponential smoothing, neural-network approaches, threshold autoregressive models, and multivariate dynamic specifications \citep{maia2008forecasting,maia2011holt,sun2018threshold,lin2016interval,han2016vector}.
Exploratory and multivariate tools such as principal component analysis and visualization have likewise been adapted to interval data \citep{lerademacher2012symbolic,wang2012cipca,zhang2022visualization}.
More recent work further addresses covariance and precision-matrix estimation for interval-valued data \citep{tian2024mcd,wu2024idgm,qin2026estimating}, as well as clustering, classification, model averaging, and machine-learning methods \citep{decarvalho2009dynamic,yang2019prediction,alcacer2024ordinal}.
Overall, this strand of research has established a broad methodological foundation for the analysis of interval-valued observations. Its primary focus, however, remains on finite-dimensional interval-valued vectors or interval-valued series observed at discrete time points, whereas settings in which intervals evolve over a continuous domain require the observations to be considered as complete trajectories.

The second strand focuses on interval-valued functional data, for which interval-valued observations are modeled as functions defined over a continuous domain.
IVFD are thus a natural extension of general interval-valued data within the functional framework.
Existing research on the modeling and analysis of IVFD has concentrated mainly on three directions.
In functional principal component analysis, methods extract the dominant modes of variation of interval-valued curves in order to reduce dimension and summarize high-dimensional functional information \citep{ikeda2010extension,sun2022time}.
In classification and clustering, the literature has studied grouping and ordinal classification for interval-valued functional observations \citep{sun2023interval,alcacer2024ordinal}.
In regression, functional linear and nonparametric procedures have been proposed to relate interval-valued functions to other variables \citep{beyaztas2022functional,nasirzadeh2021linear,nasirzadeh2022nonparametric}.
Intraday prices of financial assets provide a representative application of IVFD: by partitioning a trading day into a continuum of ordered intraday windows and forming an interval from the low and high prices within each window, each trading day can be represented as an interval-valued price curve, thereby describing the continuous joint evolution of the price level and within-window fluctuation over the session.

High-quality data are essential for reliable statistical modeling \citep{rousseeuw1987robust}, parameter estimation \citep{gervini2012outlier}, and predictive analysis \citep{hyndman2007robust}, as outlying observations can substantially distort both the analytical procedures and their results. Identifying and handling outliers before formal modeling is therefore a crucial step in maintaining data quality and ensuring the validity of subsequent analyses.
For classical functional data, outlier detection has been studied extensively, two prominent families that are particularly relevant to the present work are depth-based and projection-based methods. Depth-based methods identify outlying curves by measuring the centrality or outlyingness of functional observations, considering features such as overall shape, departure direction, and local variation \citep{lopez2009concept,sun2011functional,dai2020functional,jimenezvaron2024pointwise}. Projection-based methods reduce functional observations to finite-dimensional representations, typically through functional principal component analysis, and detect outliers in the resulting score space using robust or multivariate criteria \citep{ren2017projection}. By utilizing information from the full functional data, both classes of methods can identify anomalies in functional level and shape.
Although outlier detection methods for point-valued functional data are relatively well developed, methods specifically designed for interval-valued functional data remain limited.

Motivated by this gap, this paper develops a robust projection-based outlier detection framework for interval-valued functional data (IVFD). Each interval-valued functional observation is represented by its center and log-radius functions, and interval-valued functional principal component analysis (IFPCA) is formulated under a weighted product-space framework to balance the scales of the two components while preserving their dependence. Because the empirical mean and projection structure estimated from the full sample may be distorted by contamination, we introduce the interval-valued least trimmed functional scores (ILTFS) method, which identifies a robust reference subset by minimizing a trimmed aggregate of standardized IFPCA score distances. Based on the ILTFS subset, the projection structure is re-estimated and the subset is further refined. The resulting projection distances are then converted into empirical $p$-values and adjusted using the Benjamini-–Hochberg procedure at a prespecified target false discovery rate, yielding the complete ILTFS-FDR outlier detection procedure. We establish the finite-sample replacement breakdown point of the ILTFS mean estimator and the descent property and convergence of its associated concentration-step algorithm. Simulation studies evaluate the detection power and error-control performance of ILTFS-FDR under different dependence structures, contamination schemes, sample sizes, and grid resolutions, while an empirical application to intraday high--low price curves of exchange-traded funds demonstrates its practical effectiveness for detecting abnormal interval-valued functional observations.

The structure of this paper is organized as follows: Section \ref{sec2} introduces the proposed outlier detection method for interval-valued functional data. Section \ref{sec3} discusses the theoretical properties of the estimator. Section \ref{sec4} evaluates the finite-sample performance through simulation studies. Section \ref{sec5} applies the proposed method to real-world applications in high-frequency ETF data demonstrate the effectiveness and robustness of the proposed method. Section \ref{sec6} concludes the paper. All theoretical proofs are provided in the Appendix \ref{appendixA}.

\section{Methodology} \label{sec2}
This section develops a robust projection-based procedure for outlier detection in interval-valued functional data.
Suppose that we observe $N$ interval-valued functional observations on a compact intraday time domain $\mathcal{T}$,
\[
\mathcal{X}_N=\left\{X_i(t)=\big[X_i^{l}(t),X_i^{u}(t)\big]: i=1,\ldots,N,\ t\in\mathcal{T}\right\},
\]
where $X_i^{l}(t)$ and $X_i^{u}(t)$ denote the lower and upper bound functions of the $i$th observation, respectively, with $X_i^{l}(t)\leq X_i^{u}(t)$ for all $t\in\mathcal{T}$. Following \cite{brito2012modelling} and \cite{zhao2023robust}, each interval-valued curve is transformed into its center and log-radius functions,
\[
X_i^{c}(t)=\frac{X_i^{u}(t)+X_i^{l}(t)}{2},\qquad
X_i^{r}(t)=\log\! \left(\frac{X_i^{u}(t)-X_i^{l}(t)}{2}\right),
\]
and is represented as a bivariate functional observation
$
\mathbf{X}_i(t)=\big(X_i^{c}(t),X_i^{r}(t)\big)^\top,\: t\in\mathcal{T}.
$
We model the vector-valued functional observations as
\begin{equation}
\mathbf{X}_i(t)
=
\begin{cases}
\boldsymbol{\mu}_i(t) + \boldsymbol{\varepsilon}_i(t),
& i \in \mathcal{O}_N, \\[2mm]
\boldsymbol{\mu}_0(t) + \boldsymbol{\varepsilon}_i(t),
& i \notin \mathcal{O}_N,
\end{cases}
\label{eq:outlier_model}
\end{equation}
where $\mathcal{O}_N$ denotes the index set of outlying samples,
$\boldsymbol{\mu}_0(t)$ represents the mean vector function of the normal
samples, where
$\boldsymbol{\mu}_0(t)
=
\bigl(\mu_0^c(t),\mu_0^r(t)\bigr)^{\top}$
characterizes the mean function, capturing both the central tendency and the
variability structure. For an outlying sample $i \in \mathcal{O}_N$, the mean
function is denoted by
$\boldsymbol{\mu}_i(t)
=
\bigl(\mu_i^c(t),\mu_i^r(t)\bigr)^{\top}$,
which deviates from $\boldsymbol{\mu}_0(t)$ in either component. The term
$\boldsymbol{\varepsilon}_i(t)
=
\bigl(\varepsilon_i^c(t),\varepsilon_i^r(t)\bigr)^{\top}$
represents a bivariate stochastic error process with
$\mathbb{E}\!\left[\boldsymbol{\varepsilon}_i(t)\right]=\mathbf{0}$.
An outlying observation may deviate from the normal pattern through the center function, the log-radius function, or their joint variation.

Direct outlier detection on interval-valued functional observations is challenging because each observation is inherently high-dimensional and contains information from both its center and  range throughout the analysis. We therefore adopt a projection-based strategy that represents the dominant variation of regular interval-valued curves in a low-dimensional subspace and identifies outliers according to their departures from this subspace. This strategy requires a dimension-reduction method capable of accommodating interval-valued functional observations. Although PCA and FPCA are widely used to extract dominant modes of variation from multivariate and functional data, its standard formulation is designed for point-valued functions and does not directly account for the joint variation and scale differences between the center and log-radius components.

Accordingly, the proposed methodology proceeds in three stages. First, we construct an IFPCA representation for the bivariate functions $\mathbf{X}_i(t)$. The weighted product-space formulation balances the scales of the center and log-radius components while preserving their dependence, thereby providing a joint low-dimensional representation of the interval-valued functional observations. Second, because the empirical mean and projection structure estimated from the full sample may be distorted by abnormal curves, we introduce the ILTFS method to identify a robust reference subset by minimizing a trimmed aggregate of standardized IFPCA score distances. Third, starting from the ILTFS subset, we re-estimate the IFPCA projection structure and refine the reference subset using a robust threshold for the projection distances. The final projection distances are then converted into empirical $p$-values, and the Benjamini--Hochberg procedure is applied at a prespecified target false discovery rate level to determine the outlier set. The complete three-stage procedure is referred to as the ILTFS-FDR outlier detection procedure.

\subsection{Interval-valued functional principal component analysis}
\label{sec2.1}
Based on the bivariate representation above, we regard the bivariate functional observation $\mathbf{X}(\cdot)$ as an element of the product Hilbert space
$
\mathcal{H}=L^2(\mathcal{T})\times L^2(\mathcal{T}).
$
A key issue is that the center and log-radius components may have different scales. To balance their contributions, we define the weighted inner product on $\mathcal{H}$ as
\[
\inner{\mathbf{f}}{\mathbf{g}}
=
w_c\langle f^c,g^c\rangle_2+w_r\langle f^r,g^r\rangle_2,
\]
where $\mathbf{f}=(f^c,f^r)^\top$, $\mathbf{g}=(g^c,g^r)^\top$, $\langle u,v\rangle_2=\int_{\mathcal{T}}u(t)v(t)\,dt$ is the standard $L^2$ inner product, and $w_c>0$ and $w_r>0$ are component-specific weights. In practice, the weights can be chosen as the inverse integrated variances of the two components, so that neither the center nor the log-radius dominates the projection merely because of scale. The induced norm is denoted by $\|\mathbf{f}\|_w^2=\inner{\mathbf{f}}{\mathbf{f}}$.
To preserve the dependence between the center and log-radius functions, we model their joint covariance structure. Let $\boldsymbol{\mu}(t)=\mathbb{E}\{\mathbf{X}(t)\}$ and define the covariance kernel
\[
\mathbf{C}(s,t)
=
\mathbb{E}\left[(\mathbf{X}(s)-\boldsymbol{\mu}(s))\otimes(\mathbf{X}(t)-\boldsymbol{\mu}(t))\right]
=
\begin{pmatrix}
	C_{cc}(s,t) & C_{cr}(s,t)\\
	C_{rc}(s,t) & C_{rr}(s,t)
\end{pmatrix},
\]
where $C_{ij}(s,t)=\mathrm{Cov}\{X^i(s),X^j(t)\}$. The associated weighted covariance operator
$\Gamma_w:\mathcal{H}\to\mathcal{H}$ is defined such that, for
$\mathbf{f}\in\mathcal{H}$, its $j$th component is
\[
(\Gamma_w\mathbf{f})^{(j)}(t)
=
\sum_{i\in\{c,r\}}
w_i
\int_{\mathcal{T}}
C_{ij}(s,t)f^{(i)}(s)\,ds
=
\left\langle\!\left\langle
\mathbf{C}_{\cdot,j}(\cdot,t),\mathbf{f}
\right\rangle\!\right\rangle_w,
\qquad j\in\{c,r\},\ t\in\mathcal{T},
\]
where $\mathbf{C}_{\cdot,j}(\cdot,t)$ denotes the vector whose
$i$th component is $C_{ij}(\cdot,t)$.

Following the multivariate FPCA framework of \cite{happ2018multivariate}, under standard regularity conditions, $\Gamma_w$ is a linear, self-adjoint, and compact operator. 
 By the Spectral Theorem, there exists a countable set of eigenfunctions $\bm{\psi}_m$ and non-negative eigenvalues $v_m$ that satisfying
$
\Gamma_w\boldsymbol{\psi}_m=v_m\boldsymbol{\psi}_m,\quad
\inner{\boldsymbol{\psi}_m}{\boldsymbol{\psi}_n}=\mathbb{1}(m=n).
$
The eigenfunctions form an orthonormal basis in the weighted space, ensuring that distinct principal components are uncorrelated.
Crucially, the eigenvalues $v_m$ quantify the amount of variation explained by each principal component. Without loss of generality we assume that the eigenvalues are ordered in a strictly non-increasing sequence: $
v_1 \ge v_2 \ge \dots \ge 0.$ 
This ordering implies that the first eigenfunction $\bm{\psi}_1$ captures the largest variance in the weighted space. The sequence $v_m$ typically decays to zero, allowing for effective dimension reduction. 
This yields the weighted multivariate Karhunen--Loève expansion
\[
\mathbf{X}_i(t)
=
\boldsymbol{\mu}(t)+\sum_{m=1}^{\infty}\rho_{im}\boldsymbol{\psi}_m(t),
\]
where the principal component score of observation $i$ on the $m$th component is
$
\rho_{im}=\inner{\mathbf{X}_i-\boldsymbol{\mu}}{\boldsymbol{\psi}_m}.
$
Truncating the expansion to the first $d$ components provides a low-dimensional score representation of the original interval-valued functional observation. Compared with analyzing the center and log-radius functions separately, this joint representation retains the cross-covariance between price level and interval width. It therefore provides the projection space on which the subsequent robust outlier detection procedure is built.

In practice, functional observations are observed on a discrete grid. Directly solving the integral eigenvalue problem for $\Gamma_w$ is computationally inconvenient. We therefore use a basis expansion approach to transform the infinite-dimensional problem into a finite-dimensional matrix eigenproblem. This reduction follows the multivariate FPCA construction in \cite{happ2018multivariate} and leads to the computational procedure summarized in Algorithm \ref{alg:fast_mfpca}.

\begin{algorithm}[H]
\caption{Fast IFPCA via Univariate Basis}
\label{alg:fast_mfpca}
\begin{algorithmic}

\State \textbf{Input:} Center and radius functions
$\{X_i^c(t),X_i^r(t)\}_{i=1}^{N}$ and positive weights $w_c,w_r$.

\State \textbf{Step 1: Univariate FPCA}
\For{$j\in\{c,r\}$}
    \State \textbf{(a)} Perform standard univariate FPCA on
    $\{X_i^{(j)}(t)\}_{i=1}^{N}$ to obtain the first $M_j$
    orthonormal eigenfunctions
    $\{\widehat{\phi}_k^{(j)}(t)\}_{k=1}^{M_j}$.
    \State \textbf{(b)} Compute the unweighted scores
    $\widehat{\xi}_{i,k}^{(j)}
    =\langle X_i^{(j)},\widehat{\phi}_k^{(j)}\rangle$,
    $i=1,\ldots,N$, $k=1,\ldots,M_j$, and form
    $\widehat{\boldsymbol{\Xi}}^{(j)}
    \in\mathbb{R}^{N\times M_j}$.
\EndFor

\State \textbf{Step 2: Joint Weighted Covariance Construction}
\State \hspace{\algorithmicindent}
\textbf{(1)} Construct the joint score matrix
$\boldsymbol{\Theta}
=[\,\widehat{\boldsymbol{\Xi}}^{(c)}
\ \widehat{\boldsymbol{\Xi}}^{(r)}\,]
\in\mathbb{R}^{N\times(M_c+M_r)}$.

\State \hspace{\algorithmicindent}
\textbf{(2)} Define
$\mathbf{D}
=\operatorname{diag}
(\underbrace{\sqrt{w_c},\ldots,\sqrt{w_c}}_{M_c\ \mathrm{times}},
\underbrace{\sqrt{w_r},\ldots,\sqrt{w_r}}_{M_r\ \mathrm{times}})$.

\State \hspace{\algorithmicindent}
\textbf{(3)} Compute
$\widehat{\mathbf{Z}}
=(N-1)^{-1}\mathbf{D}\boldsymbol{\Theta}^{\top}
\boldsymbol{\Theta}\mathbf{D}$.

\State \textbf{Step 3: Eigen-Decomposition}
\State \hspace{\algorithmicindent}
\textbf{(1)} Solve
$\widehat{\mathbf{Z}}\mathbf{c}_m
=\widehat{\nu}_m\mathbf{c}_m$ subject to
$\|\mathbf{c}_m\|=1$, $m=1,\ldots,M$.

\State \hspace{\algorithmicindent}
\textbf{(2)} Retain the eigenvalues
$\{\widehat{\nu}_m\}_{m=1}^{M}$ and the corresponding eigenvectors
$\{\mathbf{c}_m\}_{m=1}^{M}$.

\State \textbf{Step 4: Functional Recovery and Scoring}
\For{$m=1,\ldots,M$}
    \State \textbf{(a)} Compute the scores
    $\widehat{\rho}_{i,m}
    =\boldsymbol{\Theta}_{i,\cdot}\mathbf{D}\mathbf{c}_m
    =\sum_{j\in\{c,r\}}\sum_{n=1}^{M_j}
    [\mathbf{c}_m]_{n}^{(j)}
    \sqrt{w_j}\,\widehat{\xi}_{i,n}^{(j)}$,
    $i=1,\ldots,N$.

    \State \textbf{(b)} For each $j\in\{c,r\}$, recover
    $\widehat{\psi}_m^{(j)}(t)
    =w_j^{-1/2}\sum_{n=1}^{M_j}
    [\mathbf{c}_m]_{n}^{(j)}
    \widehat{\phi}_n^{(j)}(t)$.

    \State \textbf{(c)} Set
    $\widehat{\boldsymbol{\psi}}_m(t)
    =(\widehat{\psi}_m^{(c)}(t),
    \widehat{\psi}_m^{(r)}(t))^{\top}$.
\EndFor

\State \textbf{Output:} The eigenvalues
$\{\widehat{\nu}_m\}_{m=1}^{M}$, scores
$\{\widehat{\rho}_{i,m}\}_{i=1,m=1}^{N,M}$, and eigenfunctions
$\{\widehat{\boldsymbol{\psi}}_m(t)\}_{m=1}^{M}$.

\end{algorithmic}
\end{algorithm}

\subsection{Robust clean-subset estimation}
\label{subsec:iltfs}

The IFPCA representation provides a joint low-dimensional representation of the
center and log-radius functions. However, the empirical mean function and
covariance operator estimated from the full sample may be severely distorted
by abnormal observations. Such contamination may alter the leading
eigenfunctions and eigenvalues, causing the estimated projection space to
partially accommodate the outlying curves and thereby leading to masking and
swamping. We therefore first seek a relatively clean subset of observations
from which the location, covariance structure, and principal projection space
can be estimated robustly.
\cite{ren2017projection} extended the least-trimmed-squares principle to
functional data by replacing squared regression residuals with standardized
squared functional principal component scores. Building on this construction,
we introduce the \emph{interval-valued least trimmed functional scores}
(ILTFS) method for interval-valued functional data. The proposed method
selects a subset of fixed size $h$ by minimizing a trimmed aggregate of
standardized IFPCA scores. 

Let
$
\mathcal{H}_{h}
=
\left\{
H\subseteq\{1,\ldots,N\}: |H|=h
\right\}
$
denote the collection of all candidate subsets of size $h$.
To define the ILTFS criterion, we first suppose that a fixed set of initial
eigenvalues and eigenfunctions
$
\left\{
\left(
\widehat{\nu}_{m}^{*},
\widehat{\boldsymbol{\phi}}_{m}^{*}
\right)
\right\}_{m=1}^{d}
$
is available, whose construction will be described below. For any nonempty subset $H\subseteq\{1,\ldots,N\}$, define the corresponding
subset mean function and the subset-centered IFPCA score of observation $i$
along the $m$th initial eigenfunction by
\[
\widehat{\boldsymbol{\mu}}_{H}(t)
=
\frac{1}{|H|}\sum_{j\in H}\mathbf{X}_{j}(t),
\quad t\in\mathcal{T},
\qquad
\widehat{\rho}_{im}(H)
=
\inner{
\mathbf{X}_{i}-\widehat{\boldsymbol{\mu}}_{H}
}{
\widehat{\boldsymbol{\phi}}_{m}^{*}
},
\quad m=1,\ldots,d.
\]
Based on these scores, the standardized IFPCA score distance of observation
$i$ relative to the candidate subset $H$ is defined as
\begin{equation}
\label{eq:iltfs_score_distance}
D_{i}(H)
=
\sum_{m=1}^{d}
\frac{\widehat{\rho}_{im}^{\,2}(H)}
{\widehat{\nu}_{m}^{*}},
\qquad i=1,\ldots,N.
\end{equation}
Although the initial eigenvalues and eigenfunctions are held fixed during
the subset search, $D_i(H)$ varies with $H$ through the subset mean function
$\widehat{\boldsymbol{\mu}}_{H}$. It therefore measures the standardized
departure of $\mathbf{X}_i$ from the location induced by $H$ in the common
IFPCA score space.
Let $D_{(1)}(H)\leq\cdots\leq D_{(N)}(H)$ denote the ordered score
distances associated with $H$. The ILTFS clean subset is defined as
\begin{equation}
\label{eq:iltfs_objective}
\widehat{H}_{\mathrm{ILTFS}}
\in
\underset{H\in\mathcal{H}_{h}}{\arg\min}
\sum_{j=1}^{h}D_{(j)}(H).
\end{equation}
For notational convenience, write
$Q(H)=\sum_{j=1}^{h}D_{(j)}(H)$ for the objective function in
\eqref{eq:iltfs_objective}. Thus,
$\widehat{H}_{\mathrm{ILTFS}}$ is the candidate subset whose induced mean
function yields the smallest trimmed total standardized score distance.
The corresponding ILTFS mean estimator is
$\widehat{\boldsymbol{\mu}}_{\mathrm{ILTFS}}(t)
=h^{-1}\sum_{i\in\widehat{H}_{\mathrm{ILTFS}}}\mathbf{X}_{i}(t)$.

It remains to specify how the fixed initial eigenvalues and eigenfunctions
used in \eqref{eq:iltfs_score_distance} are obtained. To compare different
candidate subsets meaningfully, their score distances must be evaluated with
respect to a common set of projection directions and standardization factors.
Accordingly, the eigenfunctions and eigenvalues used during the subset search
are estimated once from a robust initial subset and are then held fixed
throughout the search. This avoids allowing each candidate subset to determine
its own coordinate system, under which the resulting distance values would
not be directly comparable.
Following \citet{TianLiPeng2026}, we first apply the interval-valued minimum diagonal product (IMDP), an interval-valued extension of the minimum diagonal product criterion, to obtain an initial subset
$H_{\mathrm{IMDP}}$. We then
perform IFPCA on the observations indexed by $H_{\mathrm{IMDP}}$ and let
$
\left\{
\left(
\widehat{\nu}_{m}^{*},
\widehat{\boldsymbol{\phi}}_{m}^{*}
\right)
\right\}_{m=1}^{d},
\:
\widehat{\nu}_{1}^{*}
\geq
\cdots
\geq
\widehat{\nu}_{d}^{*}
>
0
$
denote the resulting initial eigenvalues and eigenfunctions. These quantities
provide common projection directions and variance-standardization factors
for evaluating all candidate subsets.

An exhaustive minimization of \eqref{eq:iltfs_objective} is computationally
infeasible, we therefore adapt the
concentration-step strategy. Given a current subset $H_k$, we compute its
subset mean $\widehat{\boldsymbol{\mu}}_{H_k}$ and the distances
$D_i(H_k)$ for all $i=1,\ldots,N$. After ordering these distances, the
indices associated with $D_{(1)}(H_k),\ldots,D_{(h)}(H_k)$ are retained to
form the updated subset $H_{k+1}$.
The concentration step is repeated until the selected subset no longer
changes. Since this greedy procedure may converge to different local minima
from different starting subsets, it is run from $M$ random initial subsets.
Each run starts from a randomly selected subset $H_{0}$ containing two
observations. Among the $M$ converged subsets, the one with the smallest value
of $Q(H)$ is retained as $\widehat{H}_{\mathrm{ILTFS}}$. The complete
procedure is summarized in Algorithm~\ref{alg:iltfs}.
\begin{algorithm}[H]
\caption{ILTFS Clean-Subset Estimation}
\label{alg:iltfs}
\begin{algorithmic}

\State \textbf{Step 1: Robust Initialization}
\State \hspace{\algorithmicindent}
\textbf{(1)} Apply the IMDP criterion to obtain an initial subset
$\widehat{H}_{\mathrm{IMDP}}$ of size $h$.
\State \hspace{\algorithmicindent}
\textbf{(2)} Perform IFPCA on $\widehat{H}_{\mathrm{IMDP}}$ to obtain the
initial eigenvalues $\{\widehat{\nu}_m^{*}\}_{m=1}^{d}$ and eigenfunctions
$\{\widehat{\boldsymbol{\phi}}_m^{*}(t)\}_{m=1}^{d}$.

\State \textbf{Step 2: ILTFS Subset Search}
\For{$r=1,\ldots,M$}
    \State \textbf{(a)} Randomly select an initial subset
    $H_0^{(r)}\subset\{1,\ldots,N\}$ with $|H_0^{(r)}|=2$,
    and set $k\gets0$.
    \Repeat
        \State \textbf{(b)} Compute the subset mean
        $\widehat{\boldsymbol{\mu}}_{H_k^{(r)}}(t)$ and the distances
        $D_i(H_k^{(r)})$, $i=1,\ldots,N$.
        \State \textbf{(c)} Retain the $h$ observations with the smallest
        values of $D_i(H_k^{(r)})$ to form the updated subset
        $H_{k+1}^{(r)}$.
        \State \textbf{(d)} Set $k\gets k+1$.
    \Until{$H_k^{(r)}=H_{k-1}^{(r)}$}
    \State \textbf{(e)} Set $\widetilde{H}_r\gets H_k^{(r)}$ and
    $Q_r\gets Q(\widetilde{H}_r)$.
\EndFor

\State \textbf{Step 3: Final Selection}
\State \hspace{\algorithmicindent}
\textbf{(1)} Set
$r^{*}\gets\arg\min_{1\leq r\leq M}Q_r$.
\State \hspace{\algorithmicindent}
\textbf{(2)} Set
$\widehat{H}_{\mathrm{ILTFS}}\gets\widetilde{H}_{r^{*}}$.

\State \textbf{Output:} The clean subset
$\widehat{H}_{\mathrm{ILTFS}}$.

\end{algorithmic}
\end{algorithm}
\subsection{Refinement and outlier detection}

The subset ${H}_{\mathrm{ILTFS}}$ is obtained through trimming and therefore provides a robust but conservative estimate of the  observations. Some normal curves may be excluded together with the contaminated observations, so its complement cannot be interpreted directly as the outlier set. To recover these observations, set $H_{0}={H}_{\mathrm{ILTFS}}$, re-estimate the IFPCA components from $H_{0}$, and compute the corresponding projection distances $D_{i}^{(0)}$.To refine the conservative subset $H_0$, we follow the robust scale construction of \cite{rousseeuw1993alternatives} and use the $S_n$ estimator to define a robust upper bound for the projection distances. Specifically, for $i,j=1,\ldots,N$, let $Z_i^{(0)}=\ln(D_i^{(0)}+\epsilon)$, $\hat{\mu}^{(0)}=\operatorname*{med}_{i}Z_i^{(0)}$, and $\hat{\sigma}^{(0)}=1.1926\,\operatorname*{med}_{i}\{\operatorname*{med}_{j}|Z_i^{(0)}-Z_j^{(0)}|\}$, where $\epsilon>0$ prevents taking the logarithm of zero. The refined subset is then defined as
\begin{equation}
	\label{eq:refined_subset}
	H_1
	=
	H_0
	\cup
	\left\{
	i:
	Z_i^{(0)}
	\leq
	\hat{\mu}^{(0)}
	+
	\gamma\hat{\sigma}^{(0)}
	\right\},
\end{equation}
where $\gamma>0$ determines the refinement threshold. This rule preserves the observations in $H_0$ and adds those whose log-distances do not exceed the robust upper bound.

Using $H_{1}$ as the reference subset, we re-estimate the IFPCA components and compute the final projection distances $D_{i}$. The refined subset is used to construct an empirical reference distribution. Specifically, the significance level associated with observation $i$ is defined by the empirical right-tail probability
\begin{equation}
	\label{eq:empirical_pvalue}
	p_{i}
	=
	\frac{
		\#\left\{
		j\in H_{1}:
		Z_{j}\geq Z_{i}
		\right\}
	}{
		|H_{1}|
	},
	\qquad i=1,\ldots,N.
\end{equation}
A smaller value of $p_{i}$ indicates that observation $i$ lies further in the upper tail of the projection-distance distribution represented by $H_{1}$.
The collection $\{p_{i}\}_{i=1}^{N}$ is treated as a multiple-testing problem, and the Benjamini--Hochberg procedure is applied at a prescribed FDR level $q\in(0,1)$ \citep{benjamini1995controlling}. Let $p_{(1)}\leq\cdots\leq p_{(N)}$ denote the ordered empirical $p$-values. The adaptive rejection threshold and the resulting outlier set are determined by
\begin{equation}
	\label{eq:bh_detection}
	\hat{k}
	=
	\max
	\left\{
	k\in\{1,\ldots,N\}:
	p_{(k)}
	\leq
	\frac{k}{N}q
	\right\},
	\qquad
	\hat{\tau}
	=
	p_{(\hat{k})},
	\qquad
	\hat{\mathcal{O}}
	=
	\left\{
	i:
	p_{i}\leq\hat{\tau}
	\right\}.
\end{equation}
If the set defining $\hat{k}$ is empty, we set $\hat{\tau}=0$. The complete refinement and detection procedure is summarized in Algorithm~\ref{alg:DetectionFDR}.
\begin{algorithm}[H]
\caption{ILTFS-FDR Outlier Detection Procedure}
\label{alg:DetectionFDR}
\begin{algorithmic}
\Require Interval-valued functional observations $\mathcal{X}_N$;
ILTFS subset $\widehat{H}_{\mathrm{ILTFS}}$;
retained dimension $d$;
refinement parameter $\gamma$;
stabilizing constant $\epsilon>0$;
target FDR level $q$
\Ensure Refined subset $H_1$ and estimated outlier set
$\widehat{\mathcal{O}}$

\State \textbf{Step 1: Subset Refinement}
\State \hspace{\algorithmicindent}\textbf{(1)} Set
$H_0\gets\widehat{H}_{\mathrm{ILTFS}}$.
\State \hspace{\algorithmicindent}\textbf{(2)} Perform IFPCA on the
observations indexed by $H_0$ and compute the projection distances
$D_i^{(0)}$, $i=1,\ldots,N$.
\State \hspace{\algorithmicindent}\textbf{(3)} Set
$Z_i^{(0)}\gets\log(D_i^{(0)}+\epsilon)$ and estimate the location
$\widehat{\mu}^{(0)}$ and scale $\widehat{\sigma}^{(0)}$ using the
median and the $S_n$ estimator, respectively.
\State \hspace{\algorithmicindent}\textbf{(4)} Construct the refined
subset $H_1$ according to \eqref{eq:refined_subset}.

\State \textbf{Step 2: Empirical \(\boldsymbol{p}\)-Values}
\State \hspace{\algorithmicindent}\textbf{(1)} Perform IFPCA on the
observations indexed by $H_1$ and compute the projection distances
$D_i$, $i=1,\ldots,N$.
\State \hspace{\algorithmicindent}\textbf{(2)} Set
$Z_i\gets\log(D_i+\epsilon)$ and compute the empirical $p$-values
$p_i$, $i=1,\ldots,N$, according to
\eqref{eq:empirical_pvalue}.

\State \textbf{Step 3: FDR-Controlled Detection}
\State \hspace{\algorithmicindent}\textbf{(1)} Order the empirical
$p$-values as $p_{(1)}\leq\cdots\leq p_{(N)}$.
\State \hspace{\algorithmicindent}\textbf{(2)} Determine
$\widehat{k}$ and the adaptive threshold $\widehat{\tau}$ according to
\eqref{eq:bh_detection}.
\State \hspace{\algorithmicindent}\textbf{(3)} Set
$\widehat{\mathcal{O}}\gets
\{i:p_i\leq\widehat{\tau}\}$.

\State \textbf{Output:} The refined subset $H_1$ and the estimated
outlier set $\widehat{\mathcal{O}}$.
\end{algorithmic}
\end{algorithm}

\section{Theoretical properties}
 \label{sec3}

This section establishes the theoretical framework for the proposed interval-valued functional outlier detection method. To verify the effectiveness and reliability of this method in handling interval-valued functional data, we systematically study the statistical robustness of the estimator and the convergence of the algorithm.
Given that outliers can severely distort statistical inference, it is crucial to assess the estimator's resistance to contamination. To rigorous evaluate the robustness of the proposed outlier detection procedure, we investigate the finite-sample breakdown point of the estimator. 
Since our outlier detection relies on the subspace spanned by the first $d$ interval-valued functional principal components, we adopt this projection-based breakdown point. 
The breakdown point of the estimator $\hat{\boldsymbol{\mu}}$ is defined as:
\begin{equation}
    b_{N, \boldsymbol{\psi}}(\hat{\boldsymbol{\mu}}, \mathcal{X}) = \min_{1 \le m \le N} \left\{ \frac{m}{N} : \max_{1 \le k \le d} \sup_{\mathcal{X}'} \|\hat{\boldsymbol{\mu}}_{\mathcal{X}} - \hat{\boldsymbol{\mu}}_{\mathcal{X}'}\|_{\boldsymbol{\psi}_k} = \infty \right\},
\end{equation}
where the supremum is taken over all possible corrupted samples $\mathcal{X}'$ obtained by replacing $m$ original observations in $\mathcal{X}$ with arbitrary interval-valued functions.
The following theorem establishes the breakdown point of the proposed ILTFS estimator. 

\begin{theorem} \label{thm:breakdown}Let $h$ be the size of the subset used in the ILTFS procedure. The finite-sample breakdown point of the proposed estimator $\hat{\boldsymbol{\mu}}_{H_{ILTFS}}$ is given by:\begin{equation}\epsilon_N = \frac{\min(N - h + 1, h)}{N}.\end{equation}\end{theorem}

\begin{remark}Theorem \ref{thm:breakdown} provides a theoretical guide for the selection of the subset size $h$. To maximize the robustness of the procedure, one typically chooses $h \approx \lfloor N/2 \rfloor + 1$. With this configuration, the breakdown point reaches asymptotically $50\%$, meaning the proposed method can withstand nearly half of the data being arbitrarily contaminated without breaking down.\end{remark}
The core logic of the proof of Theorem \ref{thm:breakdown} relies on comparing the objective function values between a clean subset and a potentially contaminated one. We demonstrate that if the number of outliers $m$ is less than $\min(N - h + 1, h)$, there exists at least one clean subset $H_1$ consisting entirely of non-outlying observations. The objective function evaluated on this clean subset provides a finite upper bound. Consequently, any subset that includes enough outliers to drive the norm of the mean estimator $\|\hat{\boldsymbol{\mu}}\|$ to infinity would result in an objective function value strictly exceeding this bound. Since the ILTFS algorithm seeks to minimize this objective function, such a contaminated subset would not be selected, thereby ensuring the boundedness of the estimator.

The proposed ILTFS procedure employs a fast iterative algorithm based on Concentration Steps, as detailed in Algorithm \ref{alg:iltfs}. A critical requirement for such an iterative scheme is its convergence that the objective function must not increase after each update.
\begin{theorem} \label{thm:convergence}Let $H_1 \subset \{1, \dots, N\}$ be an arbitrary subset of size $h$. Let $D_{(1)}(H_1) \le \dots \le D_{(N)}(H_1)$ denote the ordered distances computed using the estimators derived from $H_1$. Construct a new subset $H_2$ consisting of the indices corresponding to the $h$ smallest distances, such that $\{D_i(H_1) : i \in H_2\} = \{D_{(1)}(H_1), \dots, D_{(h)}(H_1)\}$. Then, the following inequality holds:\begin{equation}\sum_{i=1}^h D_{(i)}(H_2) \le \sum_{i=1}^h D_{(i)}(H_1),\end{equation}where $D_{(i)}(H_2)$ are the ordered distances computed using the updated estimators derived from $H_2$.\end{theorem}

The proof relies on decomposing the update process into two stages. First, by the definition of order statistics, selecting the observations with the smallest distances naturally minimizes the sum of distances for the given estimators. Second, re-estimating the mean and covariance parameters based on the new subset $H_2$ further reduces the sum of squared distances within $H_2$.

\section{Simulation studies} \label{sec4}
In this section, we evaluate the finite-sample performance of the proposed outlier detection procedure. The simulation design follows the general framework of \cite{ren2017projection}, while extending it from point-valued functional observations to interval-valued functional data. Based on this, the simulation study examines whether the proposed procedure can detect abnormal behavior arising from their joint evolution under different temporal dependence structures, contamination compositions, sample sizes, and grid resolutions.
For each interval-valued functional observations, we first generate a latent center process $X_i^c(t)$ and a latent log-radius process $X_i^r(t)$. Specifically, we set $X_i^c(t)=\mu_c(t)+\varepsilon_i^c(t),\: X_i^r(t)=\mu_r(t)+\varepsilon_i^r(t).$
Without loss of generality, we set $\mu_c(t)=\mu_r(t)=0$.

Dependence between the center and log-radius components is introduced through
$\varepsilon_i^c(t)=e_i^c(t)+\beta z_i(t)$ and $\varepsilon_i^r(t)=e_i^r(t)+\beta z_i(t)$,where $e_i^c(t)$ and $e_i^r(t)$ are component-specific stochastic processes, $z_i(t)$ is a process shared by the two components, and $\beta$ controls the strength of their dependence. 
All stochastic processes are generated on the equally spaced grid $t_j=j/p$, $j=1,\ldots,p$. We consider three data-generating processes (DGPs): Brownian motion (BM), an autoregressive process (AR), and a moving average process (MA). \begin{itemize}\item \textit{DGP1 (BM)} the process follows Brownian motion with independent increments satisfying $e_i(t_{j+1})-e_i(t_j)\sim\mathcal{N}(0,0.2^2)$\item \textit{DGP2 (AR)} the process follows the AR(2) recursion $e_i(t_j)=e_i(t_{j-1})-0.9e_i(t_{j-2})+u_i(t_j)$, where $u_i(t_j)\sim\mathcal{N}(0,1)$\item \textit{DGP3 (MA)} the process follows the MA(2) specification $e_i(t_j)=u_i(t_j)+0.5u_i(t_{j-1})+0.3u_i(t_{j-2})$, with $u_i(t_j)\sim\mathcal{N}(0,1)$ \end{itemize}
Under each DGP, three mutually independent stochastic processes, \(e_i^c(t)\), \(e_i^r(t)\), and \(z_i(t)\), are generated according to the same underlying mechanism, representing the center-specific, radius-specific, and shared components, respectively. These three mechanisms represent distinct forms of temporal dependence and therefore allow us to assess whether the detection performance is sensitive to the underlying stochastic dynamics.

Contaminated samples are generated by combining two outlier mechanisms. Let $\rho$ denote the contamination proportion and let $\mathcal{O}_N$ be the set of outlying observations. \textit{Shape Outlier} is produced by adding a localized sinusoidal perturbation to the center function over $[1/3,1/2]$,
$
    \mu_i^c(t)
    =\mu_0^c(t)
    +\gamma\sin(2\pi t)I_{[1/3,\,1/2]}(t),
    \: i\in\mathcal{O}_N,
$
while simultaneously enlarging the radius by a factor of $1.5$ over the same interval.
\textit{Magnitude Outlier} is produced by adding a localized linearly increasing perturbation to the center function over a randomly selected interval $[t_a,t_b]$,
$
    \mu_i^c(t)
    =\mu_0^c(t)
    +\gamma t I_{[t_a,\,t_b]}(t),
    \: i\in\mathcal{O}_N,
$
while simultaneously introducing an additive enlargement of the radius over the same interval,
$
   \mu_i^r(t)
    =\mu_0^r(t)
    +0.2\gamma I_{[t_a,\,t_b]}(t),
    \: i\in\mathcal{O}_N,
$
where $t_a$ and $t_b$ are sampled from the observation grid and $\gamma$ controls the contamination intensity.
We then consider two mixture scenarios. \textit{Case I} contains 75\% shape outliers and 25\% magnitude outliers, whereas \textit{Case II} reverses these proportions. The comparison between the two cases is intended to determine whether localized shape distortions and level-shift deviations present different levels of detection difficulty.

The simulation settings are $N\in\{100,200,500\}$, $p\in\{300,500\}$, and $\rho\in\{0.02,0.04,0.10,0.20\}$. Unless otherwise stated, we set the outlier intensity to $\gamma=3.0$ and the nominal false discovery rate to $q=0.05$. Each configuration is independently replicated 100 times, and the reported results are Monte Carlo averages. Detection performance is evaluated using the true positive rate (TPR),
false positive rate (FPR), and false discovery rate (FDR), defined as
$$
\mathrm{TPR}
=\frac{|\widehat{\mathcal{O}}_N\cap\mathcal{O}_N|}
{|\mathcal{O}_N|},
\qquad
\mathrm{FPR}
=\frac{|\widehat{\mathcal{O}}_N\setminus\mathcal{O}_N|}
{N-|\mathcal{O}_N|},
\qquad
\mathrm{FDR}
=\frac{|\widehat{\mathcal{O}}_N\setminus\mathcal{O}_N|}
{\max\{|\widehat{\mathcal{O}}_N|,1\}}.
$$
Here, $\widehat{\mathcal{O}}_N$ denotes the estimated outlier set.
TPR measures the proportion of true outliers that are correctly
identified, FPR measures the proportion of normal observations that are
incorrectly flagged as outliers, and FDR measures the proportion of false
discoveries among all observations declared anomalous.

Figures~\ref{fig:simulation_composite_p300} and~\ref{fig:simulation_composite_p500} summarize the main detection results for $p=300$ and $p=500$, respectively. Across the three stochastic processes, the proposed method generally attains high detection power. The strongest and most stable performance is observed under the AR process, for which the TPR remains close to or above 0.95 in most configurations. By comparison, the BM and MA settings exhibit greater sensitivity to the contamination proportion and the mixture composition. In particular, Case II is generally more difficult than Case I, indicating that the randomly located level-shift outliers are harder to distinguish than the localized shape perturbations under the current simulation design. This difference is most visible in configurations with smaller samples, a coarser observation grid, or a relatively large proportion of magnitude outliers.

The comparison between the two grid resolutions shows that increasing $p$ from 300 to 500 generally stabilizes the detection performance, especially under BM and MA dependence. A denser grid provides more information about local departures and therefore improves the representation of short-lived anomalous features. Larger sample sizes also tend to reduce finite-sample variability. Nevertheless, the influence of the contamination proportion is not uniform across the three DGPs. Very sparse contamination makes FDR control difficult because even a small number of false positives can account for a large fraction of all discoveries. At the other extreme, under BM dependence, a high contamination proportion can reduce TPR when a more conservative threshold is used, suggesting that dense contamination may affect the robust scale calibration. Thus, sparse and dense contamination create different finite-sample challenges: the former primarily affects the stability of FDR, whereas the latter may reduce power in the more difficult dependence settings.

Table~\ref{tab:performance_k} summarizes the effect of different robust scale multipliers $\gamma$ on the experimental results, with $N$ fixed at 100 and $p$ at 500. The results indicate that $\gamma=1$ is highly liberal, as it maintains an average TPR close to 0.98 while producing FPR values ranging from approximately 2.5\% to 8.3\% and empirical FDR levels that substantially exceed the nominal level in most configurations. Increasing the multiplier to $\gamma=2$ reduces the FPR to below 0.7\% in all reported settings while retaining high power under AR and MA dependence. However, when $\rho=0.02$, the empirical FDR remains between approximately 0.170 and 0.219 because the number of true outliers is very small. Once $\rho$ increases to 0.10 or 0.20, the FDR for $\gamma=2$ is generally below 0.025, demonstrating effective error control outside the extremely sparse regime.
A further comparison across the different $\gamma$ settings reveals a clear power--conservativeness trade-off. The choice $\gamma=3$ nearly eliminates false discoveries, but the gain in error control can be accompanied by a substantial loss of power. This loss is particularly pronounced under BM and MA dependence at larger contamination proportions; for example, when $\rho=0.20$, the TPR under BM decreases to 0.520 in Case I and 0.480 in Case II. In contrast, the AR process remains comparatively insensitive to the multiplier, with TPR values above 0.92 in all settings and above 0.95 in most cases. These findings indicate that the effect of $\gamma$ depends not only on the desired degree of FDR control but also on the temporal dependence structure. Based on these results, $\gamma=2$ provides the most balanced performance in the reported experiments, whereas $\gamma=3$ is appropriate only when avoiding false discoveries is considerably more important than maintaining detection power.

Overall, the simulation result support three main conclusions. First, ILTFS-FDR can identify both shape and magnitude outliers in interval-valued functional data across a broad range of finite-sample settings. Second, detection is generally easier under AR dependence and for the shape-dominated Case I, while BM dependence and magnitude-dominated contamination are more challenging. Third, the robust scale multiplier governs a meaningful trade-off between power and false-discovery control: moderate calibration yields strong overall performance, whereas overly liberal or overly conservative choices can respectively inflate FDR or reduce TPR. These result validate the effectiveness of the proposed method while also identifying the sparse-contamination regime and complex dependence structures as the principal finite-sample limitations.

\begin{figure}[htbp]
    \centering
    \includegraphics[width=0.80\textwidth]{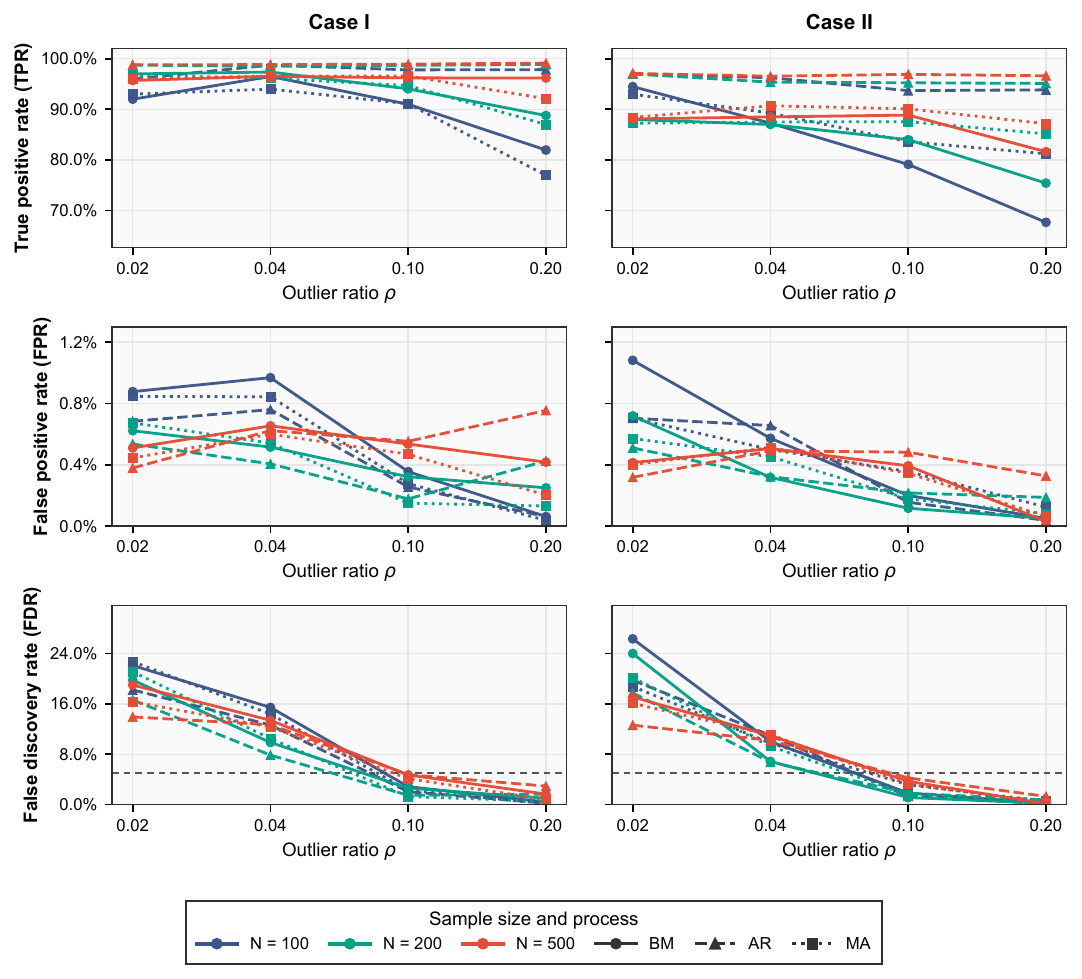}
    \caption{Detection performance of the proposed method under the fixed grid resolution $p=300$.}
    \label{fig:simulation_composite_p300}
\end{figure}

\begin{figure}[htbp]
    \centering
    \includegraphics[width=0.80\textwidth]{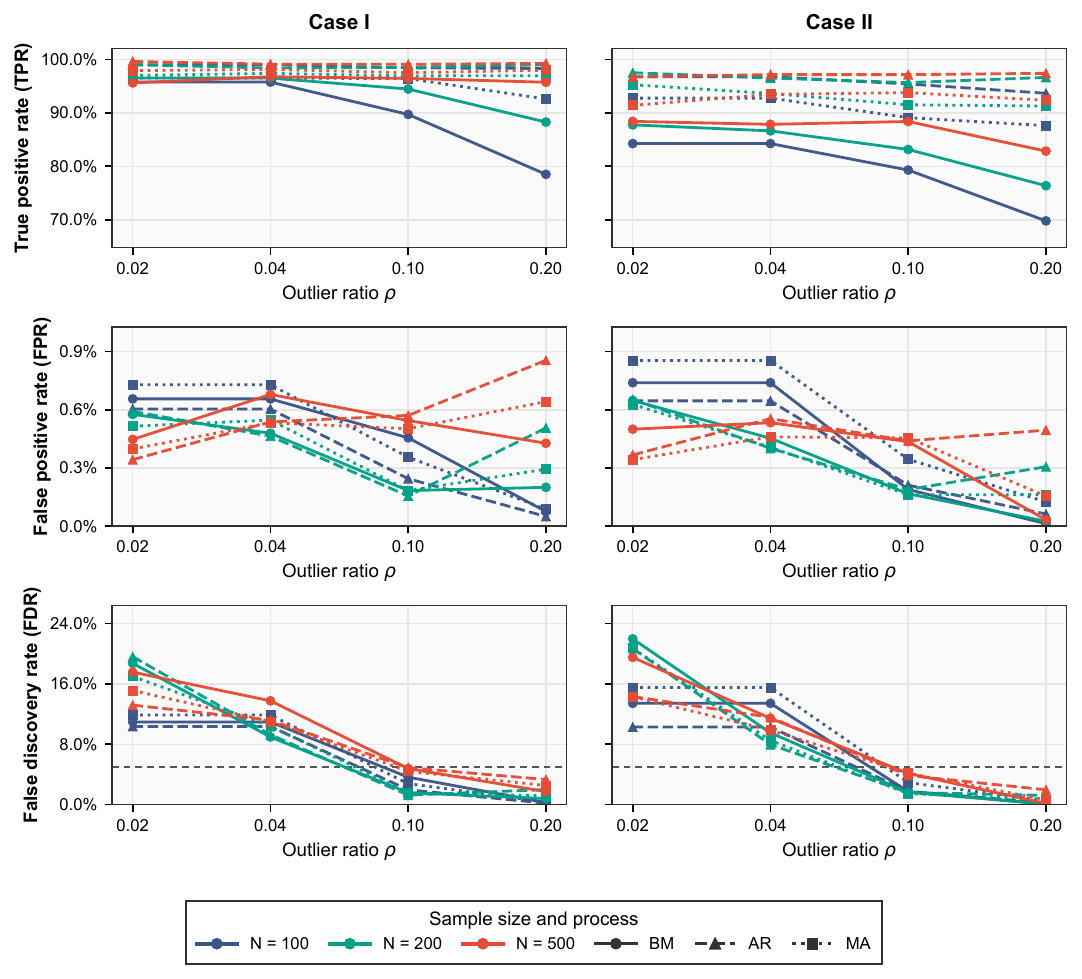}
    \caption{Detection performance of the proposed method under the fixed grid resolution $p=500$.}
    \label{fig:simulation_composite_p500}
\end{figure}

\begin{table}[htbp]
\centering
\caption{Detection performance across contamination proportions $\rho$ and robust scale multipliers $\gamma$, and the FDR level is $q=0.05$.}
\label{tab:performance_k}
\resizebox{\textwidth}{!}{
\begin{tabular}{cllccccccccc}
\toprule
\multirow{2}{*}{$\rho$} & \multirow{2}{*}{process} & \multirow{2}{*}{case} & \multicolumn{3}{c}{$\gamma=1$} & \multicolumn{3}{c}{$\gamma=2$} & \multicolumn{3}{c}{$\gamma=3$} \\
\cmidrule(lr){4-6} \cmidrule(lr){7-9} \cmidrule(lr){10-12}
 & & & TPR & FPR & FDR & TPR & FPR & FDR & TPR & FPR & FDR \\
\midrule
\multirow{6}{*}{$0.02$}
& BM & Case I & 0.978 & 0.079 & 0.743 & 0.965 & 0.006 & 0.187 & 0.930 & 0.000 & 0.014 \\
& & Case II & 0.915 & 0.078 & 0.753 & 0.878 & 0.007 & 0.219 & 0.845 & 0.000 & 0.014 \\
& AR & Case I & 0.993 & 0.064 & 0.705 & 0.990 & 0.006 & 0.196 & 0.985 & 0.000 & 0.000 \\
& & Case II & 0.973 & 0.066 & 0.711 & 0.975 & 0.007 & 0.207 & 0.955 & 0.000 & 0.002 \\
& MA & Case I & 0.978 & 0.054 & 0.667 & 0.970 & 0.005 & 0.170 & 0.965 & 0.000 & 0.004 \\
& & Case II & 0.943 & 0.051 & 0.657 & 0.953 & 0.006 & 0.207 & 0.885 & 0.000 & 0.005 \\
\midrule
\multirow{6}{*}{$0.04$}
& BM & Case I & 0.984 & 0.083 & 0.628 & 0.965 & 0.005 & 0.090 & 0.923 & 0.000 & 0.003 \\
& & Case II & 0.915 & 0.078 & 0.617 & 0.866 & 0.005 & 0.095 & 0.786 & 0.000 & 0.004 \\
& AR & Case I & 0.996 & 0.075 & 0.590 & 0.984 & 0.005 & 0.093 & 0.988 & 0.000 & 0.000 \\
& & Case II & 0.974 & 0.068 & 0.567 & 0.965 & 0.004 & 0.079 & 0.956 & 0.000 & 0.000 \\
& MA & Case I & 0.980 & 0.059 & 0.547 & 0.974 & 0.006 & 0.104 & 0.953 & 0.000 & 0.002 \\
& & Case II & 0.925 & 0.056 & 0.538 & 0.936 & 0.004 & 0.085 & 0.891 & 0.000 & 0.001 \\
\midrule
\multirow{6}{*}{$0.10$}
& BM & Case I & 0.977 & 0.072 & 0.388 & 0.945 & 0.002 & 0.016 & 0.887 & 0.000 & 0.001 \\
& & Case II & 0.936 & 0.066 & 0.372 & 0.832 & 0.002 & 0.016 & 0.685 & 0.000 & 0.000 \\
& AR & Case I & 0.992 & 0.060 & 0.338 & 0.985 & 0.002 & 0.013 & 0.982 & 0.000 & 0.000 \\
& & Case II & 0.972 & 0.058 & 0.331 & 0.957 & 0.002 & 0.016 & 0.937 & 0.000 & 0.000 \\
& MA & Case I & 0.980 & 0.055 & 0.318 & 0.969 & 0.002 & 0.016 & 0.915 & 0.000 & 0.000 \\
& & Case II & 0.940 & 0.048 & 0.300 & 0.915 & 0.002 & 0.015 & 0.858 & 0.000 & 0.000 \\
\midrule
\multirow{6}{*}{$0.20$}
& BM & Case I & 0.973 & 0.033 & 0.117 & 0.883 & 0.002 & 0.008 & 0.520 & 0.000 & 0.002 \\
& & Case II & 0.903 & 0.025 & 0.096 & 0.764 & 0.000 & 0.001 & 0.480 & 0.000 & 0.000 \\
& AR & Case I & 0.993 & 0.032 & 0.112 & 0.990 & 0.005 & 0.020 & 0.976 & 0.002 & 0.010 \\
& & Case II & 0.980 & 0.029 & 0.103 & 0.966 & 0.003 & 0.012 & 0.922 & 0.000 & 0.001 \\
& MA & Case I & 0.977 & 0.026 & 0.094 & 0.970 & 0.003 & 0.012 & 0.596 & 0.000 & 0.001 \\
& & Case II & 0.934 & 0.025 & 0.093 & 0.913 & 0.002 & 0.007 & 0.727 & 0.000 & 0.000 \\
\bottomrule
\end{tabular}
}
\end{table}

\section{Real data application} \label{sec5}
Financial markets are continually influenced by information arrivals\citep{andersen2003micro}, liquidity conditions\citep{amihud2002illiquidity}, and investor sentiment \citep{baker2006investor}.
Their interaction at particular times can trigger abrupt intraday price jumps or sharp increases in volatility
\citep{lee2008jumps, liu2015does}.
Such abnormal fluctuations can materially impair asset valuation, portfolio construction, and the effectiveness of trading and hedging decisions, thereby posing substantial challenges for financial risk assessment and management.
\citep{mackinlay1997event,geboers2023review}.
Accordingly, the reliable identification of abnormal trading days is
therefore important for both market surveillance and financial risk
management.

Motivated by this consideration, we apply the proposed outlier detection procedure to
high-frequency exchange-traded fund (ETF) data. The detected anomalies are then used to construct a simple downside-risk
control strategy as an additional assessment of their practical relevance. We
consider four representative ETFs from the Resset database: the SSE 50 ETF and
CSI 300 ETF, which represent large-capitalization equities and the broader
equity market in China, respectively, and the Nasdaq 100 ETF and S\&P 500 ETF,
which provide exposure to the U.S. technology sector and the broader U.S.
equity market. For each ETF, the dataset contains five-minute high and low
prices together with daily opening and closing prices.
ETF high-frequency data are naturally suited to interval-valued functional
modeling. Within each five-minute trading interval, both the lowest and
highest transaction prices are directly observed, so the intraday price record
contains an interval rather than a single representative value. Conventional
point-valued analyses typically retain only the closing price, midpoint, or
one of the interval boundaries. Although such representations preserve the
evolution of the price level, they discard the within-interval trading range.
By jointly using the high and low prices, the interval-valued representation
retains information about both intraday price movements and local price
dispersion, allowing an abnormal trading day to be characterized through
either component or through their joint evolution.
Let $P_t^{\ell}(v)$ and $P_t^{u}(v)$ denote the intraday low- and high-price
functions on trading day $t$, where $v\in[0,1]$ is standardized trading time,
and let $P_t^{o}$ denote the daily opening price. After removing trading days
with excessive missing records and interpolating the remaining missing values
onto a common intraday grid, we normalize the two price functions by the
opening price and define
\[
    X_t(v)
    =
    \left[
    X_t^{\ell}(v),X_t^{u}(v)
    \right]
    =
    \left[
    \frac{P_t^{\ell}(v)}{P_t^{o}},
    \frac{P_t^{u}(v)}{P_t^{o}}
    \right],
    \qquad v\in[0,1].
\]
This normalization removes overnight price gaps and reduces scale differences
across trading days and ETFs. 

To ensure that the empirical analysis uses only historically observable
information, ILTFS-FDR is implemented through a rolling window. At the end of
day $t-1$, define the reference sample
$
    \mathcal{X}_{t-1}^{(W)}
    =
    \{X_s:s=t-W,\ldots,t-1\},
$
where $W$ is the window length. Applying ILTFS-FDR to
$\mathcal{X}_{t-1}^{(W)}$ yields an estimated set of abnormal trading-day indices,
$
    \widehat{\mathcal{O}}_{t-1}^{(W)}
    \subseteq
    \{t-W,\ldots,t-1\}.
$
Because positive and negative anomalies have different implications for risk
management, we combine the outlier classification with the sign of the most
recent daily return. Let
\[
r_{t-1}
=
\frac{P_{t-1}^{c}-P_{t-1}^{o}}{P_{t-1}^{o}},
\qquad
S_{t-1}
=
\begin{cases}
1, & t-1\in\widehat{\mathcal{O}}_{t-1}^{(W)}
     \ \text{and}\ r_{t-1}<0,\\
0, & \text{otherwise},
\end{cases}
\]
where $r_{t-1}$ denotes the open-to-close return on day $t-1$,
$P_{t-1}^{c}$ is the closing price, and $S_{t-1}$ is the resulting
downside-risk alert.
Thus, an ILTFS-FDR outlier generates a trading alert only when it is
accompanied by a negative return. When $S_{t-1}=1$, the ETF position is closed
at the market open on day $t$; otherwise, the strategy remains invested. The
resulting position indicator is $a_t=1-S_{t-1}$, and the gross strategy return
is
$
    r_t^{\mathrm{i}}=a_t r_t.
$
The position is re-established at the first subsequent market open for which
the lagged alert equals zero. This construction treats ILTFS-FDR as a defensive
exposure-timing signal rather than as an asset-selection rule and avoids the
use of future information.


Table~\ref{tab:real_detection} summarizes the resulting downside-risk alerts.
ILTFS-FDR produces between 10 and 24 alerts for each ETF, corresponding to only
0.85\%--2.03\% of the evaluation sample. These alerts form 8--20 distinct
episodes, and the longest episode lasts only two or three consecutive trading
days. The procedure therefore identifies a sparse set of short-lived abnormal
periods rather than persistently classifying large portions of the sample as
risky.
\begin{table}[htbp]
\centering
\caption{Downside-risk alerts in the ETF application}
\label{tab:real_detection}
\resizebox{\textwidth}{!}{
\begin{tabular}{lrrrrrr}
\toprule
ETF
& trading days
& alerts
& alert rate
& episodes
& absolute-return ratio
& bottom-decile share\\
\midrule
SSE 50 ETF     & 1,180 & 24 & 2.03 & 20 & 2.20 & 41.7 \\
CSI 300 ETF    & 1,180 & 18 & 1.53 & 15 & 2.03 & 44.4 \\
Nasdaq 100 ETF & 1,179 & 10 & 0.85 &  8 & 3.20 & 50.0 \\
S\&P 500 ETF   & 1,177 & 16 & 1.36 & 15 & 1.65 & 37.5 \\
\bottomrule
\end{tabular}
}
\begin{minipage}{\textwidth}\footnotesize
\textit{Note:} An episode is a maximal sequence of consecutive alert days.
The absolute-return ratio is the mean absolute daily return on alert days
divided by that on non-alert days. The bottom-decile share is the proportion
of alert days whose returns fall in the lowest decile of the corresponding
ETF's unconditional daily-return distribution; the unconditional benchmark is
10\%.
\end{minipage}
\end{table}
The alerts are economically distinct from ordinary trading days. Mean absolute
returns on alert days are 1.65--3.20 times those on non-alert days. Moreover,
37.5\%--50.0\% of the alerts fall in the bottom decile of the corresponding
return distribution, compared with an unconditional benchmark of 10\%.
Downside-tail observations are therefore overrepresented among the alerts by
factors of 3.75--5.00. These results indicate that the procedure does not
merely detect statistically unusual curves; after incorporating the return
sign, the resulting alerts are strongly associated with economically adverse
price movements.
The cross-market comparison further shows that the detections contain both
asset-specific and common information. Among the 1,176 dates shared by all
four ETFs, 1,124 contain no alert, 38 contain an alert for exactly one ETF, 12
contain alerts for two ETFs, and only two contain alerts for three ETFs. No
date is simultaneously flagged for all four series. Most alerts are therefore
series-specific, whereas a smaller subset reflects disturbances shared across
related markets. In particular, three ETFs are simultaneously flagged on
March 19, 2020, and August 5, 2024.

Figure~\ref{fig:etf_strategy_performance} presents the out-of-sample cumulative
net asset values of the downside-risk-control strategy and the continuously
invested benchmark. The strategy achieves a higher terminal cumulative return
for all four ETFs. For the SSE 50 ETF and CSI 300 ETF, the terminal returns
increase from 16.03\% to 25.48\% and from 18.30\% to 27.58\%, respectively.
For the Nasdaq 100 ETF and S\&P 500 ETF, they increase from 285.26\% to
308.81\% and from 156.89\% to 187.77\%, respectively. Because alerts account
for no more than 2.03\% of the sample, these improvements are not obtained by
remaining outside the market for extended periods. Instead, they suggest that
a small number of abnormal negative-return days make a disproportionate
contribution to adverse investment performance.

\begin{figure}[htbp]
\centering
\includegraphics[width=\textwidth]{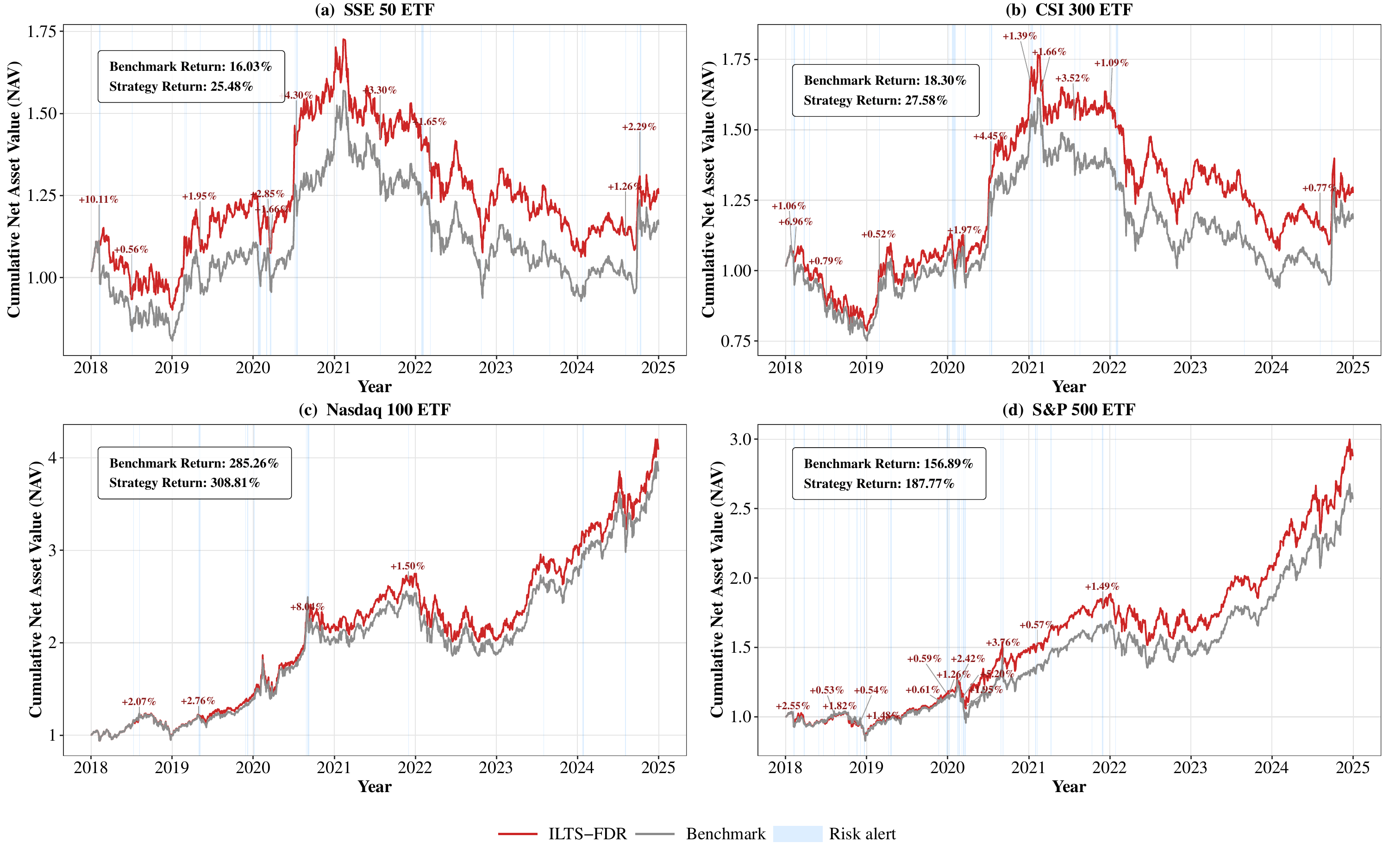}
\caption{Out-of-sample cumulative net asset values of the ILTFS-FDR
downside-risk-control strategy and the continuously invested benchmark. The
panels correspond to the SSE 50 ETF, CSI 300 ETF, Nasdaq 100 ETF, and
S\&P 500 ETF, respectively. Highlighted dates indicate downside-risk alerts
constructed from lagged information.}
\label{fig:etf_strategy_performance}
\end{figure}

Overall, the empirical evidence supports the practical relevance of ILTFS-FDR
for abnormal trading-day detection. Statistically, the method identifies a
small set of interval-valued functional observations associated with unusually
large price movements and a pronounced concentration in the lower return
tail. Economically, converting these detections into a lagged defensive signal
improves terminal out-of-sample valued across all four ETFs. The results
therefore show that the proposed interval-valued functional framework can
provide useful information for both abnormal-market monitoring and dynamic
downside-risk control.
\section{Conclusions} \label{sec6}
This paper develops a robust projection-based framework for outlier detection
in interval-valued functional data. The proposed framework first represents
each interval-valued functional observation through its center and log-radius
functions and applies IFPCA to obtain a joint low-dimensional representation.
The ILTFS method is then introduced to identify a robust reference subset by
minimizing a trimmed aggregate of standardized IFPCA score distances. After
refining this subset, the projection structure is re-estimated, and the
resulting distances are converted into empirical $p$-values and adjusted using
the Benjamini--Hochberg procedure at a prespecified target false discovery
rate level. We further derive the finite-sample replacement breakdown point of
the ILTFS mean estimator and establish the descent property of the associated
concentration-step algorithm.

Simulation studies show that the proposed framework can effectively detect
both shape and magnitude outliers under different dependence structures,
sample sizes, grid resolutions, and contamination settings. Its performance
generally becomes more stable with larger samples and denser observation
grids, although extremely sparse contamination and more complex dependence
structures remain challenging. The application to high-frequency ETF data
identifies a small number of  abnormal trading periods that are
associated with unusually large price movements and a greater concentration
in the lower tail of the return distribution. The resulting lagged
downside-risk signals also improve terminal cumulative returns relative to
the continuously invested benchmarks for the four ETFs considered.

\section*{Declaration of competing interest}
The authors report there are no competing interests to declare.

\section*{Acknowledgments}
\addcontentsline{toc}{section}{Acknowledgments}
The research work described in this paper was supported by the National Natural Science Foundation of China (Nos. 72471010, W2511079 and 72071008).

\appendix
\section{Proofs for results}\label{appendixA}

\begin{theorem}[Finite-Sample Breakdown Point]
The finite-sample breakdown point of the proposed estimator is $\epsilon_N = \min(N - h + 1, h) / N$.
\end{theorem}

\begin{proof}
We first show that if the number of outliers $m < \min(N - h + 1, h)$, the estimator remains bounded.
Since $m < N - h + 1$, there exists at least one subset $H_1 \subset \{1, \dots, N\}$ of size $h$ consisting entirely of original observations. Let $M_1=\max_{i=1,...,N}||X_i||_w^2<\infty$.
We explicitly derive the upper bound for the objective function of this clean subset $H_1$. By definition:
\begin{align*}
    \sum_{i \in H_1} D_i(H_1) &= \sum_{i \in H_1} \sum_{k=1}^d \hat{\lambda}_{k, H_1}^{-1} \inner{\bX_i - \hat{\bmu}_{H_1}}{\hat{\bpsi}_{k, H_1}}^2 \nonumber \\
    &\le \sum_{i \in H_1} \sum_{k=1}^d \hat{\lambda}_{k, H_1}^{-1} \norm{\bX_i - \hat{\bmu}_{H_1}}^2 \nonumber \\
    &\le \sum_{i \in H_1} \sum_{k=1}^d \hat{\lambda}_{k, H_1}^{-1} 2\left( \norm{\bX_i}^2 + \norm{\hat{\bmu}_{H_1}}^2 \right) \nonumber \\
    &\le 4h M_1 \sum_{k=1}^d \hat{\lambda}_{k, H_1}^{-1}
\end{align*}
This confirms that the minimum objective function value is bounded.

Suppose that there exists a principal component $l$ such that $|\langle\langle\hat{\mu}_{H_{opt}},\hat{\psi}_{l,opt}\rangle\rangle_w|=M$ for some large $M$. Let $H_{opt}$ be the optimal subset derived by the algorithm. Since $m < h$, $H_{opt}$ contains a non-empty subset $J_0$ of original observations. 
We analyze the objective function value for $H_{opt}$:
\begin{align*}
    \sum_{i \in H_{opt}} D_i(H_{opt}) &= \sum_{i \in H_{opt}} \sum_{k=1}^d \hat{\lambda}_{k, opt}^{-1} \inner{\bX_i - \hat{\bmu}_{H_{opt}}}{\hat{\bpsi}_{k, opt}}^2 \nonumber \\
    &\ge \sum_{k=1}^d \hat{\lambda}_{k, opt}^{-1} \sum_{i \in J_0} \inner{(\bX_i - \hat{\bmu}_{J_0}) + (\hat{\bmu}_{J_0} - \hat{\bmu}_{H_{opt}})}{\hat{\bpsi}_{k, opt}}^2 \nonumber \\
    &= \sum_{k=1}^d \hat{\lambda}_{k, opt}^{-1} \left[ \sum_{i \in J_0} \inner{\bX_i - \hat{\bmu}_{J_0}}{\hat{\bpsi}_{k, opt}}^2 + |J_0| \inner{\hat{\bmu}_{J_0} - \hat{\bmu}_{H_{opt}}}{\hat{\bpsi}_{k, opt}}^2 \right].
\end{align*}
Since all terms in the summation over $k$ are non-negative, the total sum is greater than or equal to any single term. Let $l \in \{1, \dots, d\}$ be such a direction where the projection captures the magnitude of the shift.
Dropping all terms where $k \neq l$ and ignoring the non-negative residual variance term, we obtain the lower bound:
\begin{align*}
\sum_{i \in H_{opt}} D_i(H_{opt}) &\ge |J_0| \hat{\lambda}_{l, opt}^{-1} \inner{\hat{\bmu}_{H_{opt}} - \hat{\bmu}_{J_0}}{\hat{\bpsi}_{l, opt}}^2\\
&\ge |J_0| \hat{\lambda}_{l, opt}^{-1} \left( \inner{\hat{\bmu}_{H_{opt}}}{\hat{\bpsi}_{l, opt}}^2 - 2 \inner{\hat{\bmu}_{H_{opt}}}{\hat{\bpsi}_{l, opt}} \inner{\hat{\bmu}_{J_0}}{\hat{\bpsi}_{l, opt}} \right) \nonumber \\
    &= \hat{\lambda}_{l, opt}^{-1} \left( |J_0| \inner{\hat{\bmu}_{H_{opt}}}{\hat{\bpsi}_{l, opt}}^2 - 2 |J_0| \inner{\hat{\bmu}_{H_{opt}}}{\hat{\bpsi}_{l, opt}} \inner{\hat{\bmu}_{J_0}}{\hat{\bpsi}_{l, opt}} \right)\\
    &\ge \hat{\lambda}_{l, opt}^{-1} \left( M^2 - 2h M M_1^{1/2} \right)
\end{align*}

This implies $\sum_{i \in H_{opt}} D_i(H_{opt}) > \sum_{i \in H_1} D_i(H_1)$ provided $M$ is sufficiently large, which contradicts the definition of $H_{opt}$. Thus, $\|\hat{\bmu}_{H_{opt}}\|_\mathbf{W}$ must be bounded.

We next to show that $\epsilon_N\le\min(N-h+1,h)/N$. If we replace $N-h+1$ data points of $\mathcal{X}$, then the optimal subset $H_{opt}$ of the contaminated dataset $\mathcal{X}'$ would contain at least one outlier. The least squares method breaks down even with one single outlier. It then follows that $||\hat{\mu}_{H_{opt}}||_w$ is not bounded because we can simply replace the observation $X_i(t)$ by $X_i'(t)=K\hat{\psi}_{l,opt}(t)$ so that $|\langle\langle{X_i',\hat{\psi}_{l,opt}}\rangle\rangle_w|=|K|$, which can be arbitrarily large. Similarly, we can easily observe that $\epsilon_N\le{h/N}$. This completes the proof.

\end{proof}

\begin{theorem}
Let $H_1\subset\{1,\dots,N\}$ be a subset of size $h$. Compute the distance $D_i(H_1)$ for $i=1,\dots,N$ based on the estimators from $H_1$. If we construct $H_2$ such that $\{D_i(H_1):i\in H_2\}=\{D_{(1)}(H_1),\dots,D_{(h)}(H_1)\}$, where $D_{(1)}(H_1)\le\dots\le D_{(h)}(H_1)$ are the ordered distances, and compute $D_i(H_2)$ based on $H_2$, then$
\sum_{i=1}^hD_{(i)}(H_2)\le\sum_{i=1}^hD_{(i)}(H_1),$
equality holds if and only if $H_1=H_2$.
\end{theorem}

\begin{proof}
First, consider the objective function for the new subset $H_2$. By the definition of order statistics (the sum of the smallest $h$ values is always less than or equal to the sum of any $h$ values), we have:
\begin{align*}\label{eq:ordering}
\sum_{i=1}^hD_{(i)}(H_2)\le\sum_{i\in H_2}D_i(H_2).
\end{align*}

Next, we analyze the relationship between the distances computed with updated parameters (from $H_2$) versus old parameters (from $H_1$). To focus on the location update, we fix the covariance components $\hat{\lambda}_{k,H_1}$ and $\hat{\bpsi}_{k,H_1}$. By expanding the squared distance $D_i(H_1)$ around the new mean $\hat{\bmu}_{H_2}$, we have:
\begin{align*}
\sum_{i\in H_2}D_i(H_1)&=\sum_{i\in H_2}\sum_{k=1}^d\frac{1}{\hat{\lambda}_{k}}{\inner{(\bX_i-\hat{\bmu}_{H_2})+(\hat{\bmu}_{H_2}-\hat{\bmu}_{H_1})}{\hat{\bpsi}_{k}}}^2\nonumber\\
&=\sum_{i\in H_2}D_i(H_2)\nonumber\\
&\quad+2\sum_{k=1}^d\frac{1}{\hat{\lambda}_{k}}\inner{\sum_{i\in H_2}(\bX_i-\hat{\bmu}_{H_2})}{\hat{\bpsi}_{k}}\inner{\hat{\bmu}_{H_2}-\hat{\bmu}_{H_1}}{\hat{\bpsi}_{k}}\nonumber\\
&\quad+h\sum_{k=1}^d\frac{1}{\hat{\lambda}_{k}}{\inner{\hat{\bmu}_{H_2}-\hat{\bmu}_{H_1}}{\hat{\bpsi}_{k}}}^2.
\end{align*}

Since $\hat{\bmu}_{H_2}=h^{-1}\sum_{i\in H_2}\bX_i$, the cross term vanishes because $\sum_{i\in H_2}(\bX_i-\hat{\bmu}_{H_2})=\mathbf{0}$. Therefore:
\begin{equation}\label{eq:ls_property}
\sum_{i\in H_2}D_i(H_1)=\sum_{i\in H_2}D_i(H_2)+h\sum_{k=1}^d\frac{1}{\hat{\lambda}_{k,H_1}}{\inner{\hat{\bmu}_{H_2}-\hat{\bmu}_{H_1}}{\hat{\bpsi}_{k,H_1}}}^2\ge\sum_{i\in H_2}D_i(H_2).
\end{equation}

Finally, by the construction of $H_2$, the set $\{D_i(H_1):i\in H_2\}$ corresponds exactly to the $h$ smallest distances based on $H_1$. Thus:
\begin{equation}\label{eq:definition}
\sum_{i\in H_2}D_i(H_1)=\sum_{i=1}^hD_{(i)}(H_1).
\end{equation}

Combining (\ref{eq:ls_property}), and (\ref{eq:definition}), we obtain the descent property:
\[
\sum_{i=1}^hD_{(i)}(H_2)\le\sum_{i\in H_2}D_i(H_2)\le\sum_{i\in H_2}D_i(H_1)=\sum_{i=1}^hD_{(i)}(H_1).
\]
The equality holds if and only if $\hat{\bmu}_{H_2}=\hat{\bmu}_{H_1}$ (and covariance estimates stabilize), which implies convergence.
\end{proof}

\bibliographystyle{cas-model2-names}
\bibliography{bib}

\end{document}